\documentclass[preprint,12pt]{elsarticle}

\usepackage[T1]{fontenc}
\usepackage[utf8]{inputenc}
\usepackage{amssymb}
\usepackage{amsmath}
\usepackage{amsthm}
\usepackage{mathtools}
\usepackage{booktabs}
\usepackage{graphicx}
\usepackage{url}
\usepackage{lineno}
\usepackage{lmodern}   
\usepackage{microtype} 

\numberwithin{equation}{section}

\theoremstyle{plain}
\newtheorem{theorem}{Theorem}[section]
\newtheorem{lemma}{Lemma}[section]

\theoremstyle{definition}

\newtheorem{assumption}{Assumption}[section]
\theoremstyle{remark}
\newtheorem{remark}{Remark}[section]

\newcommand{\vect}[1]{\mathbf{#1}}
\newcommand{\mat}[1]{\mathbf{#1}}
\newcommand{\R}{\mathbb{R}}
\newcommand{\Expect}[1]{\langle #1 \rangle}  
\newcommand{\Ham}{\mathcal{H}}
\newcommand{\obs}{\mathrm{obs}}
\newcommand{\hid}{\mathrm{hid}}
\newcommand{\iu}{\mathrm{i}}         

\begin{document}

\begin{frontmatter}

\title{PRISM: Positive-Real Identification of Sparse Mori--Hamiltonians
       from Partial Observations\tnoteref{t1}}
\tnotetext[t1]{Extended and corrected version of arXiv:2606.15505v1; submitted to Automatica.}

\author{Mohammad A. Ayoubi}
\ead{maayoubi@scu.edu}
\address{Department of Mechanical Engineering, Santa Clara University,
         Santa Clara, CA 95053, USA}

\begin{abstract}
Discovering the governing equations of a physical system from data is a central goal across the sciences, yet in most experiments only a few
states are accessible while the rest stay hidden.  Existing approaches treat this partial observability as an obstacle to be removed
by first reconstructing the hidden state---a step that is ill-posed under noise and that discards the physical constraints, such as energy
conservation, that the true dynamics obey.  We show that for conservative (Hamiltonian) systems no reconstruction is needed: projecting the dynamics 
onto the measured coordinates yields a memory kernel that we prove to be a lossless positive-real rational matrix, whose poles are the hidden natural
frequencies and whose positive-semidefinite residues encode the couplings. From this kernel we recover a closed, interpretable governing equation for the
observed dynamics---identified from output data alone, passive by construction, and validated by out-of-sample forecasting.  Under stronger
conditions---an equipartitioned measure with position coupling, or a forced input--output experiment---the bare hidden frequencies and couplings of the
underlying Hamiltonian are additionally recoverable.  We test the method on linear, nonlinear, and chaotic systems under realistic noise.  Because it
returns energy-conserving equations of motion from partial measurements, it offers a common tool for problems spanning mechanics, fluid and plasma physics,
and beyond.
\end{abstract}


\end{frontmatter}


\section{Introduction}
\label{sec:intro}

Inferring the equations that govern a physical system from measured data
is a problem common to mechanics, fluid and plasma physics, chemistry, and
beyond.  A recurring difficulty unites these fields: the measured
variables are almost never the complete state.  A gyroscope reports a few
body rates, not the sloshing fuel that perturbs them; a wall probe records
a field at one location, not the hidden modes driving it. The measurement
is thus a low-dimensional projection of a larger, partly hidden system,
and the central question is whether the full governing equation can be
recovered from that projection alone.

Data-driven equation discovery has advanced rapidly, but largely for the
fully observed case.  The Sparse Identification of Nonlinear Dynamics
(SINDy)~\cite{brunton2016discovering} extracts parsimonious models from
full-state trajectories; Koopman operator
methods~\cite{koopman1931hamiltonian,williams2015data} linearize the
dynamics in a lifted space; and delay embeddings~\cite{takens1981detecting},
combined with neural-network coordinate learning~\cite{bakarji2023discovering},
extend these ideas to partial observations. What these approaches share is a strategy:
treat partial observability as a nuisance to be removed by first
reconstructing the hidden state, then identifying the dynamics. This
two-stage route compounds difficulties. Selecting the embedding dimension
and time delay from noisy data is itself an unsolved problem~\cite{casdagli1991state};
the reconstructed coordinates carry no guarantee of energy conservation or
passivity~\cite{gao2024sindy,brunton2017chaos,arbabi2017ergodic}; and the
auxiliary embedding adds dimensions without adding physics.

Here we show that for conservative systems the reconstruction step can be
bypassed entirely.  The Mori--Zwanzig (MZ)
projection~\cite{mori1965transport,zwanzig1973nonlinear,chorin2000optimal}
maps the exact Hamiltonian flow onto the measured coordinates, producing a
closed second-order equation whose \emph{memory kernel} carries all
feedback from the hidden degrees of freedom. Our central result is a structural
theorem: when the hidden modes couple through position, this kernel is a
lossless positive-real rational matrix whose poles are the hidden frequencies
and whose positive-semidefinite residues encode the couplings.  From it we
recover a closed, interpretable governing equation for the observed
dynamics---directly from the output autocorrelation, without state
reconstruction and without a trained network, and passive by construction.
The bare hidden Hamiltonian follows under stronger conditions, which we make
precise---and only for hidden modes that couple to the observable strongly
enough to leave a fingerprint on the measured signal, the observability
condition that no reconstruction from partial data can escape.  We call the
resulting method PRISM (Positive-Real Identification of Sparse
Mori-Hamiltonians).

\subsection{Related work}
\label{sec:related}

A closely related precedent is the data-driven rational parameterization
of memory kernels in molecular dynamics by Lei, Baker, and
Li~\cite{lei2016data}.  PRISM differs in three respects.  It requires only
the measured time series, not the coarse-grained force that
Ref.~\cite{lei2016data} also needs and that is inaccessible in most
physical settings (gyroscopes do not measure force; strain gauges do not
measure pressure).  It operates on deterministic Hamiltonian trajectories,
recovering the hidden frequencies from a single excited trajectory and the
couplings from an equipartitioned ensemble or a sufficiently rich
multi-output measurement---a strictly weaker requirement than an explicit
hidden force.  And it derives the positive-real kernel structure as a
theorem from Hamiltonian physics, rather than positing it as an ansatz,
yielding formal identifiability and uniqueness guarantees.
Data-driven Mori--Zwanzig methods that learn the memory kernel from a
reduced set of observables via Koopman
regression~\cite{lin2021data,lin2022data} share PRISM's goal of recovering
reduced dynamics from partial measurements, but they \emph{fit} the kernel
as a black-box reduced-order model for general (typically dissipative)
systems; PRISM instead \emph{proves} the kernel form from conservative
structure and returns interpretable physical parameters---the hidden
frequencies and couplings---rather than a predictive surrogate.

PRISM also belongs to a controls tradition we have so far left implicit:
identifying passive, positive-real, and port-Hamiltonian systems from data.
Benner, Goyal, and Van Dooren fit a port-Hamiltonian model to
frequency-response data~\cite{benner2020identification}, Schwerdtner and Voigt
carry this to noisy measurements~\cite{schwerdtner2021port}, and
passivity-preserving reduction by interpolation at spectral
zeros~\cite{sorensen2005passivity,antoulas2005new} reaches the same guarantees
from a different direction.  PRISM wants what they want---a physically
consistent, passive model---but differs in three ways.  It treats the
\emph{lossless} conservative case, where the kernel is a reactance rather than
a general positive-real function and its poles are the hidden natural
frequencies, not merely interpolation data.  It works from the \emph{output}
autocorrelation alone, with no measured input.  And it returns an interpretable
Hamiltonian, with explicit observed--hidden couplings, in place of a black-box
realization.  The overlap is not only formal: when the hidden modes couple
reactively---as fuel slosh and gyroscopes do---the bare frequencies are exactly
the transmission (spectral) zeros of the observable receptance, and recovering
them calls for the very frequency-response route these works take.  That is
where PRISM meets the port-Hamiltonian identification literature.  Two
concurrent preprints are close in scope: one treats partial observation of
linear systems via the Mori--Zwanzig formalism~\cite{wang2026partial}, the
other learns latent port-Hamiltonian dynamics from partial
observations~\cite{li2026identify}.  PRISM differs in its closed-form reactance
kernel and its explicit account of what is recoverable for conservative
systems.

In the more challenging quantum many-body setting---operator
reconstruction on an exponentially large Hilbert space---related work
learns Hamiltonians by augmenting the dynamics with a neural network and
distilling an interpretable operator via curriculum
learning~\cite{heightman2025solving}.  PRISM targets conservative classical
systems, where the positive-real kernel yields the Hamiltonian in closed
form, with uniqueness and passivity guarantees and without a trained
network.


\section{Framework}
\label{sec:framework}

We consider a system
\begin{equation}
  \dot{\vect{z}}(t) = \vect{f}(\vect{z}(t)),
  \qquad
  \vect{y}(t) = \mat{\Gamma}\,\vect{z}(t) + \vect{n}(t),
  \label{eq:general_system}
\end{equation}
with full state $\vect{z} \!\in\! \R^{2n_\obs + 2n_\hid}$, measured output
$\vect{y}(t) \!\in\! \R^{n_\obs}$ ($n_\obs \!\ll\! n_\obs\!+\!n_\hid$),
and noise $\vect{n}(t)$, with $\mat{\Gamma}$ the output matrix that selects the
measured coordinates; only $\vect{y}$ is accessible.  We assume
$\vect{f} = \mathbf{J}\,\nabla_{\vect{z}}\Ham$ (Hamilton's equations,
symplectic matrix $\mathbf{J}$) and partition
$\vect{z} = (\vect{q}_\obs, \vect{p}_\obs, \{q_i\}, \{p_i\})$ into
observable and hidden coordinates.  The observable position is
$\vect{x}(t)\triangleq\vect{q}_\obs(t)$, with
$\dot{\vect{x}}=\mat{M}^{-1}\vect{p}_\obs$ ($\mat{M}=\mat{I}$ without loss
of generality) and $\vect{y}=\vect{x}+\vect{n}$; the $q_i,p_i$ are hidden.
The full Hamiltonian is
\begin{equation}
  \Ham = \Ham_\obs(\vect{x},\dot{\vect{x}})
    + \sum_{i=1}^{n_\hid}\!\Big(\tfrac{1}{2}p_i^2
      + \tfrac{1}{2}\omega_i^2 q_i^2\Big)
    + \Ham_c(\vect{x},\vect{q}_\hid),
  \label{eq:full_ham}
\end{equation}
where $\Ham_\obs\!=\!\tfrac{1}{2}\dot{\vect{x}}^\top\dot{\vect{x}}
+ V(\vect{x})$ is the observable self-dynamics---$V(\vect{x})$ denoting the
potential energy of the observable subsystem---and the $i$-th hidden mode
is a harmonic oscillator at frequency $\omega_i>0$.

\begin{assumption}[Linear coupling]\label{ass:A}
The coupling is linear in the hidden coordinates,
$\Ham_c = \sum_{i}\vect{x}^\top \vect{c}_i\,q_i$ with constant
vectors $\vect{c}_i \in \R^{n_\obs}$.
\end{assumption}

Projecting onto the observable subspace ($\mathcal{P}$, with
$\mathcal{Q}=I-\mathcal{P}$) yields the \emph{exact} second-order
Mori--Zwanzig equation~\cite{mori1965transport,chorin2000optimal}
\begin{equation}
  \ddot{\vect{x}}(t) = \vect{g}(\vect{x},\dot{\vect{x}})
  + \int_0^t \mat{K}(t\!-\!s)\,\vect{x}(s)\,ds
  + \vect{F}(t),
\label{eq:MZ}
\end{equation}
with self-force $\vect{g}=-\nabla_{\vect{x}}V$, memory kernel $\mat{K}(t)$
encoding the hidden-mode feedback, and fluctuation force $\vect{F}(t)$
satisfying $\mathcal{P}\vect{F}(t)=0$.  No closure approximation is made.
Throughout, a hat denotes the Laplace transform of the corresponding
time-domain quantity, with $s\in\mathbb{C}$ the complex Laplace variable,
e.g.\ $\hat{\mat{K}}(s)=\mathcal{L}\{\mat{K}(t)\}$ and
$\hat{\mat{C}}(s)=\mathcal{L}\{\mat{C}(\tau)\}$.

\section{Kernel structure}
\label{sec:kernel}

Our central result characterizes the memory kernel exactly.

\begin{theorem}[Kernel structure]\label{thm:kernel_structure}
Under Assumption~\ref{ass:A} the kernel transform $\hat{\mat{K}}(s)$ has the
partial-fraction form
\begin{equation}
  \hat{\mat{K}}(s) = \sum_{i=1}^{n_\hid}
    \frac{\mat{R}_i}{s^2 + \omega_i^2},
  \qquad \mat{R}_i = \vect{c}_i\vect{c}_i^\top \succeq 0,
  \label{eq:kernel_pf}
\end{equation}
with imaginary-axis poles and residues
$\mathrm{Res}_{s=\iu\omega_i}\hat{\mat{K}}(s) = \mat{R}_i/(2\iu\omega_i)$.
Equivalently $\hat{\mat{K}}(s)$ is a \emph{lossless positive-real}
(reactance) matrix, and the spectral density
\begin{equation}
  \mat{S}(\omega)=\sum_i\mat{R}_i\,\delta(\omega-\omega_i)
  \label{eq:spectral_density}
\end{equation}
is a positive matrix measure on $(0,\infty)$.
\end{theorem}

\noindent The proof---variation of parameters on the hidden oscillators,
then Laplace transform---is given in \ref{app:proof1}.  Because the hidden
modes are undamped, the poles lie exactly on the imaginary axis rather than
in the open left half-plane---the lossless (reactance) case, as opposed to
the dissipative kernels of the general
GLE~\cite{anderson1967network}.  The poles of $\hat{\mat{K}}$ are
thus the hidden natural frequencies: \emph{the spectrum of the hidden
Hamiltonian is read directly from the memory kernel, with no state
reconstruction.}

\section{Identifiability and reconstruction}
\label{sec:identifiability}

The kernel, and hence the Hamiltonian, can be recovered from the measured
autocorrelation.

\begin{theorem}[Identifiability]\label{thm:identifiability}
Let $\mat{C}(\tau)=\Expect{\vect{x}(t)\vect{x}(t+\tau)^\top}$ be the
observable autocorrelation, where $\Expect{\cdot}$ denotes the ensemble
(equivalently, stationary time) average.  The \emph{closed-loop spectrum}
$\{\Omega_j\}$---the roots of
$\det(s^2\mat{I}-\mat{G}-\hat{\mat{K}}(s))=0$---and the associated
mode-shape directions are identifiable, directly from the cosine content of
$\mat{C}(\tau)$ (\ref{app:proof2}), from any persistently exciting
measurement---one that excites all hidden modes---provided the observability
Gramian
\begin{equation}
  \mat{W}_o = \int_0^\infty e^{\mat{A}^\top t}\mat{\Gamma}^\top\mat{\Gamma}\,
             e^{\mat{A}t}\,dt \succ 0
  \label{eq:gramian}
\end{equation}
is positive definite, where $\mat{A}$ is the state matrix of the observable
subsystem.  The \emph{bare hidden frequencies} $\{\omega_i\}$ and coupling
vectors of Theorem~\ref{thm:kernel_structure} are a distinct, generally
non-coincident set---since the poles of $\hat{\mat{K}}(s)$ are generically
\emph{zeros} of $\hat{\mat{C}}(s)$ rather than poles of it---and are
recoverable only via the explicit deconvolution of
Eq.~\eqref{eq:kernel_from_corr}, not by direct spectral fitting of
$\mat{C}(\tau)$.  The
residue \emph{magnitudes} are additionally identifiable when the measure is
\emph{equipartitioned} (equal energy per hidden mode; canonical or
microcanonical equilibrium, or the empirical average over an ensemble of
excited trajectories), which removes the per-mode energy weighting
(see \ref{app:proof2}).
\end{theorem}

\begin{remark}[Anti-resonance mechanism]
Near a pole $s=\iu\omega_k$ of $\hat{\mat{K}}$, the
inversion~\eqref{eq:kernel_from_corr} forces $\hat{\mat{C}}(\iu\omega_k)\to\mat{0}$:
a bare hidden frequency appears as a \emph{zero} (anti-resonance), not a peak,
of the observable spectrum---the dynamic-vibration-absorber effect.  The peaks
of $\mat{C}(\tau)$ are therefore the closed-loop normal-mode frequencies
$\{\Omega_j\}$, generically distinct from the bare $\{\omega_i\}$.
\end{remark}

\noindent The generalized fluctuation--dissipation
theorem~\cite{kubo1966fluctuation} applied to the deterministic system
gives the exact inversion (derivation in \ref{app:proof2})
\begin{equation}
  \hat{\mat{K}}(s) = \bigl[s^2\hat{\mat{C}}(s)
    - s\mat{C}(0) - \mat{\Lambda}\mat{C}(0)
    - \mat{G}\hat{\mat{C}}(s)\bigr]\hat{\mat{C}}^{-1}(s),
  \label{eq:kernel_from_corr}
\end{equation}
with streaming matrix $\mat{\Lambda}=\dot{\mat{C}}(0)\mat{C}^{-1}(0)$ and
observable stiffness $\mat{G} = -\nabla^2_{\vect{x}}V\big|_{\vect{x}_0}$
($\mat{G}=\mat{0}$ for free or gyroscopic observables); $\hat{\mat{C}}(s)$
is invertible for $\mathrm{Re}(s)>0$ iff $\mat{W}_o\succ 0$.  Equipartition
is essential for the residues: a single non-equilibrium trajectory weights
each $\mat{R}_i$ by an arbitrary modal energy $E_i$, recovering the poles
but not the residue scale, whereas equipartition ($E_i\equiv E$) restores
$\mat{R}_i$ up to one global constant fixed by $\mat{C}(0)$.  Given
$\{(\omega_i,\mat{R}_i)\}$, the couplings follow by Cholesky factorization
$\vect{c}_i=\mathrm{chol}(\mat{R}_i)$ (unique up to sign), and the recovered
Hamiltonian
\begin{equation}
  \Ham = \Ham_\obs
  + \sum_{i=1}^{n_\hid}\!\Bigl(\tfrac{1}{2}p_i^2
    + \tfrac{1}{2}\omega_i^2 q_i^2\Bigr)
  + \sum_{i=1}^{n_\hid}\vect{q}_\obs^\top\vect{c}_i\,q_i
\label{eq:reconstructed_ham}
\end{equation}
is unique in minimal order by the matrix-valued Bochner
theorem~\cite{bochner1932vorlesungen,gantmacher1959}.
\emph{Assembly via Phases~1--2 (direct cosine fit of $\mat{C}(\tau)$)
recovers the closed-loop equation of motion~\eqref{eq:MZ} exactly;
recovering the bare hidden-mode Hamiltonian~\eqref{eq:full_ham} additionally
requires the deconvolution route of Eq.~\eqref{eq:kernel_from_corr}, which
Phases~1--2 do not perform.}

\section{Algorithm}
\label{sec:algorithm}

The procedure converts an observable time series
$\{\vect{x}(t_k)\}_{k=1}^M$ into the full Hamiltonian~\eqref{eq:reconstructed_ham}.

\paragraph{Phase~1 --- Kernel extraction}
Compute the empirical autocorrelation:
\begin{equation}
  \mat{C}(\tau)
  = \frac{1}{T-\tau}\!\int_0^{T-\tau}\!\vect{x}(t\!+\!\tau)\vect{x}(t)^\top dt.
  \label{eq:autocorr}
\end{equation}
Estimate the streaming matrix $\mat{\Lambda}=
\dot{\mat{C}}(0)\mat{C}^{-1}(0)$.  Two routes then recover
the hidden-mode parameters.  The kernel is related to the data by the
exact frequency-domain inversion~\eqref{eq:kernel_from_corr}, a
\emph{deconvolution}: $\hat{\mat{K}}(s)$ follows from dividing by
$\hat{\mat{C}}(s)$ in the transform domain, with
$\mat{K}(\tau)=\mathcal{L}^{-1}\{\hat{\mat{K}}(s)\}$ (\ref{app:fdt}).  This
recovers the kernel itself, but requires forming $\hat{\mat{C}}(s)$ and its
inverse from noisy data.  Because that inversion amplifies measurement
noise, we do \emph{not} compute the kernel directly to locate modes in
practice.  Instead we exploit the fact that, by
Theorem~\ref{thm:identifiability}, the support of the spectral density
of the autocorrelation is exactly the set of recoverable \emph{closed-loop}
frequencies $\{\Omega_j\}$, so that under equipartition $\mat{C}(\tau)$ is a
sum of undamped cosines at $\{\Omega_j\}$,
$\mat{C}(\tau)=\sum_j \mat{B}_j\cos(\Omega_j\tau)$.  We therefore
fit $\mat{C}(\tau)$ directly on a frequency grid by non-negative
least squares (NNLS), reading the active frequencies $\{\omega_i\}$ from
the support of the recovered weights.  This needs no
numerical differentiation of the data and is the route used for all
benchmarks below.  The explicit kernel
deconvolution~\eqref{eq:kernel_from_corr} is invoked only if the kernel
itself (rather than the pole-residue parameters) is explicitly required.

\paragraph{Phase~2 --- Sparse rational identification}
On the frequency grid $\{\nu_j\}_{j=1}^{n_g}\subset
[\nu_\mathrm{min},\nu_\mathrm{max}]$, the practical realization
used for all benchmarks fits the autocorrelation directly by
non-negative least squares,
\begin{equation}
  \min_{\{w_j\ge 0\}}
  \;\Big\|\mat{C}(\tau)
  - \sum_j w_j\cos(\nu_j\tau)\Big\|_2^2,
  \label{eq:nnls}
\end{equation}
solved channel-wise (and on the matrix entries $\hat C_{ab}$ for
$n_\obs>1$); the active set $\mathcal{I}=\{j:w_j>\epsilon_\mathrm{tol}\}$
gives the frequencies, and the per-mode residue matrices
$\mat{R}_i$ follow from the (de-tilted) cosine amplitudes of the matrix
autocorrelation.  Non-negativity of the weights is the scalar shadow of
the $\mat{R}_i\succeq0$ passivity constraint and suffices when the
residues are rank one (single hidden mode per frequency).  When
higher-rank or strongly overlapping residues must be resolved while
enforcing positive semidefiniteness exactly, Eq.~\eqref{eq:nnls}
generalizes to the semidefinite program
\begin{equation}
  \min_{\{\mat{R}_j\!\succeq\,0\}}
  \;\Big\|\hat{\mat{K}}(s)
  - \sum_j \frac{\mat{R}_j}{s^2+\nu_j^2}\Big\|_F^2
  + \lambda\sum_j\|\mat{R}_j\|_*,
  \label{eq:sdp}
\end{equation}
with $\|\cdot\|_*$ the nuclear norm (convex relaxation of the
$\ell_0$ sparsity penalty~\cite{fazel2002matrix}) and the
$\mat{R}_j\succeq0$ constraint enforced by semidefinite
programming (CVXPY\slash SCS~\cite{diamond2016cvxpy}).  The benchmarks here are all in the
rank-one regime, so the NNLS form~\eqref{eq:nnls} is used throughout.

\paragraph{Phase~3 --- Hamiltonian assembly}
For each $i\in\mathcal{I}$: compute $\vect{c}_i=\mathrm{chol}(\mat{R}_i)$,
set $\omega_i=\Omega_i$, and assemble
Hamiltonian~\eqref{eq:reconstructed_ham}.


\section{Validation}
\label{sec:validation}

We validate PRISM on three canonical Hamiltonian benchmarks, chosen so
that each isolates a distinct property of the method rather than merely
repeating a success. \emph{Benchmark~1} (linear spring chain) tests the
regime where the theory is \emph{exact at any amplitude}: a linear
conservative system with harmonic hidden modes, the setting of
Theorem~\ref{thm:identifiability} (closed-loop spectrum and mode-shape
recovery); Theorem~\ref{thm:kernel_structure}'s bare hidden-Hamiltonian claim
is not separately exercised in this benchmark, since the deconvolution step
(Eq.~\eqref{eq:kernel_from_corr}) is not applied here.
\emph{Benchmarks~2 and 3} (FPUT and H\'enon-Heiles) are
\emph{nonlinear}: the linear-kernel structure of
Assumption~\ref{ass:A} holds only in the small-amplitude limit, and
these cases test what happens as that limit is left. For them the relevant
claim is not exact recovery but \emph{graceful degradation}: PRISM
recovers the small-amplitude (effective) frequencies, and the breakdown of
that description---energy spreading in FPUT, forecast divergence at the
onset of chaos in H\'enon-Heiles---is itself a usable diagnostic. We
therefore report frequencies and a diagnostic for B2/B3, and full
frequency-plus-residue recovery only for the exactly linear B1. For all
tests, only the \emph{observable} coordinate time
series is provided to the algorithm; hidden states are withheld entirely.

\subsection{Benchmark 1: Linear spring-mass chain}
\label{sec:b1}

\paragraph{System}
$n_\mathrm{tot}=4$ equal masses connected by identical springs
($k=1$, $m=1$) with \emph{fixed--free} boundary conditions (left end
anchored to a wall, right end free), and the two end masses observed
($n_\obs=2$, $\vect{x}=[q_1,q_4]^\top$; $n_\hid=2$).
The full Hamiltonian is:
\begin{equation}
  \Ham = \sum_{i=1}^4\frac{p_i^2}{2m}
    + \frac{k}{2}q_1^2 + \sum_{i=1}^3\frac{k}{2}(q_{i+1}-q_i)^2.
  \label{eq:chain}
\end{equation}
The fixed--free spectrum has no rigid-body zero mode, so all four
normal-mode frequencies $\omega_j = 2\sin\!\big((2j-1)\pi/18\big)$,
$j=1,\ldots,4$, are oscillatory targets:
$\{0.347,1.000,1.532,1.879\}$.

\begin{figure}[t]
\centering
\includegraphics[width=\linewidth]{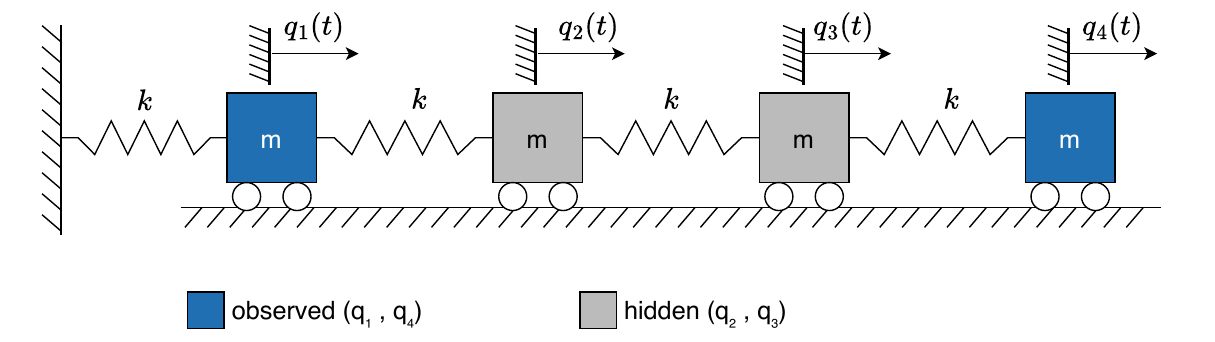}
\caption{\textbf{Benchmark~1 setup: fixed--free spring-mass chain.} Four equal
masses ($m$) joined by identical springs ($k$); the left end is anchored to a
wall, the right end is free.  The two end masses $q_1,q_4$ are observed; the
interior masses $q_2,q_3$ are hidden.  Each $q_i$ is the horizontal
displacement of mass $i$ from its equilibrium (dashed), positive to the
right.\label{fig:chain_setup}}
\end{figure}

\paragraph{Result}
From the two observed channels under 10\% additive Gaussian noise,
PRISM recovers all four \emph{closed-loop normal-mode frequencies of the
full chain} with relative error at or below $0.22\%$
(three of four at or below $0.05\%$); frequency recovery is insensitive to the
modal energy distribution, consistent with
Theorem~\ref{thm:identifiability}.  A model fit on the first $40\%$ of the
record forecasts the held-out remainder of both observed channels
(Fig.~\ref{fig:chain}a,b)---an equation of motion integrated forward, not a
spectral fit.
With the autocorrelation averaged over an equipartitioned ensemble of
$n_\mathrm{ens}=300$ initial conditions (so the residue hypothesis of
Theorem~\ref{thm:identifiability} holds), the $2\times2$ residue matrices
are recovered---including the off-diagonal coupling structure and its
sign---with a mean Frobenius error of $3.3\%$ (maximum $5.4\%$) up to the
single global scale that is intrinsically unidentifiable
(Fig.~\ref{fig:chain}c,d, comparing the recovered diagonal $R_{11}$,
off-diagonal $R_{12}$, and the full $2\times2$ residues per mode against the
analytical $\vect{v}_i\vect{v}_i^\top$). The recovered residues remain positive
semidefinite by construction, so the Cholesky factorization yields valid
coupling vectors. The
larger residue error relative to the frequency error is expected: the
off-diagonal cross-correlation $C_{12}$ has lower signal-to-noise than the
diagonal terms, and the residue scale is recovered only in the ensemble
mean, whereas the pole locations are fixed by phase and are therefore far
more robust.

\subsection{Benchmark 2: Fermi-Pasta-Ulam-Tsingou (FPUT) lattice}
\label{sec:b2}

\paragraph{System}
The $\alpha$-FPUT chain~\cite{fermi1955studies} with $n_\mathrm{tot}=8$
particles and nonlinear coupling:
\begin{equation}
  \Ham = \sum_{i=1}^8\frac{p_i^2}{2}
    + \sum_{i=0}^8\Bigl[\frac{(q_{i+1}-q_i)^2}{2}
      + \frac{\alpha(q_{i+1}-q_i)^3}{3}\Bigr],
  \label{eq:fput}
\end{equation}
with $\alpha=0.25$, fixed boundary conditions $q_0=q_9=0$.
Only the \emph{first particle} $q_1(t)$ is observed ($n_\obs=1$,
$n_\hid=7$).  Assumption~\ref{ass:A} is \emph{intentionally
violated} here, testing PRISM robustness to nonlinear hidden coupling.

\paragraph{Result}
The initial condition excites a single normal mode (mode~1), so in the
small-amplitude limit ($A=0.1$, where Assumption~\ref{ass:A}
approximately holds) the observable $q_1(t)$ shows essentially one spectral
peak and PRISM recovers that effective frequency to $0.5\%$
($\hat\omega_1=0.346$ vs.\ $0.347$). As the amplitude grows, the cubic
coupling transfers energy to other modes---the historical FPUT
energy-sharing phenomenon---and PRISM detects a growing number of active
spectral peaks: one for $A\le0.2$, two by $A=0.3$, and three by $A\ge0.7$
(Fig.~\ref{fig:fput}). We report this active-peak count versus
amplitude as a direct, data-driven diagnostic of nonlinear energy
spreading; we do not report residues, since under nonlinear coupling the
modes are not exact and a linear-residue claim would be ill-defined.

\subsection{Benchmark 3: H\'{e}non-Heiles system}
\label{sec:b3}

\paragraph{System}
The 2D H\'{e}non-Heiles Hamiltonian~\cite{henon1964applicability}:
\begin{equation}
  \Ham = \frac{1}{2}(p_x^2+p_y^2)
    + \frac{1}{2}(x^2+y^2)
    + x^2 y - \frac{y^3}{3},
  \label{eq:henon}
\end{equation}
with only $x(t)$ observed ($n_\obs=1$; hidden: $y(t)$).
The system is quasi-periodic for $E<1/6$ and chaotic above the
escape energy $E=1/6$~\cite{henon1964applicability}.

\paragraph{Result}
In the quasi-periodic regime ($E=0.05$, well below the escape energy
$E=1/6$), PRISM recovers a single dominant frequency
$\hat\omega=0.996$ from the $x(t)$ observable---the softened linearized
frequency of the well (linear value $\omega=1$). The decisive test is
forecasting: a model built from the recovered frequency, fit on the first
40\% of the record, accurately predicts the held-out remainder
(Fig.~\ref{fig:henon}c), confirming that PRISM captures the regular
dynamics and not merely its power spectrum. In the chaotic regime
($E=0.15$, $E/E_\mathrm{esc}=0.9$) the same procedure tracks for a short
horizon and then diverges from the truth (Fig.~\ref{fig:henon}d), as
required by sensitive dependence on initial conditions: no finite-mode model
can forecast a chaotic trajectory, so the forecast-breakdown horizon is
itself a data-driven diagnostic of chaos rather than a deficiency of the
method. We report frequencies and this forecast diagnostic only, not
residues, since well-defined effective modes exist only in the
small-amplitude (regular) regime.

\begin{figure}[t]
\centering
\includegraphics[width=\linewidth]{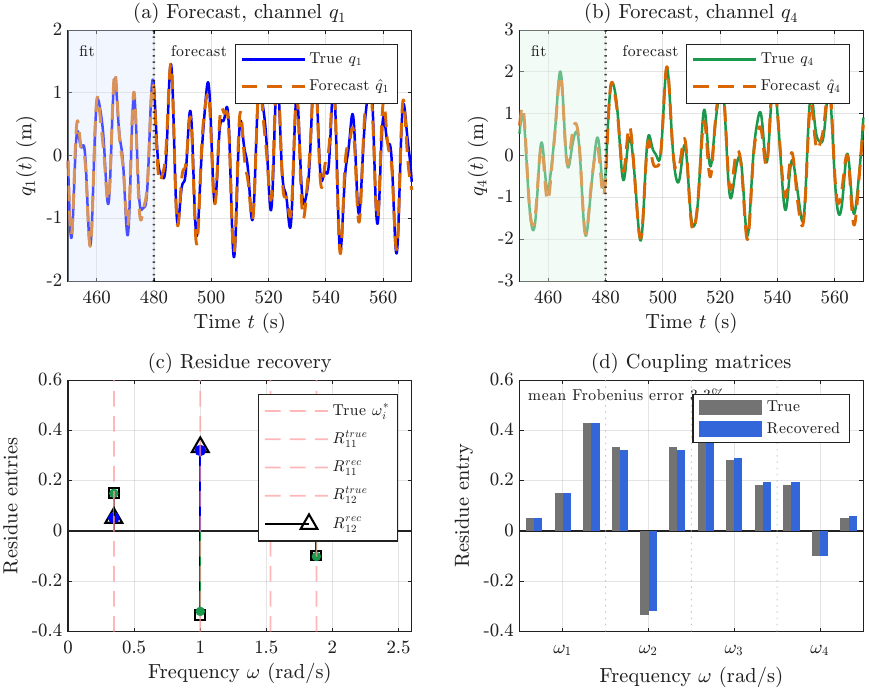}
\caption{PRISM on the fixed--free spring-mass chain ($n_\mathrm{tot}=4$,
  $\vect{x}=[q_1,q_4]^\top$; $10\%$ additive Gaussian noise).
  \textbf{(a,b)} out-of-sample forecasts of the observed channels $q_1,q_4$;
  \textbf{(c,d)} recovered vs.\ true residue and coupling matrices at the four
  identified frequencies.  See text for details.\label{fig:chain}}
\end{figure}

\begin{figure}[t]
\centering
\includegraphics[width=\linewidth]{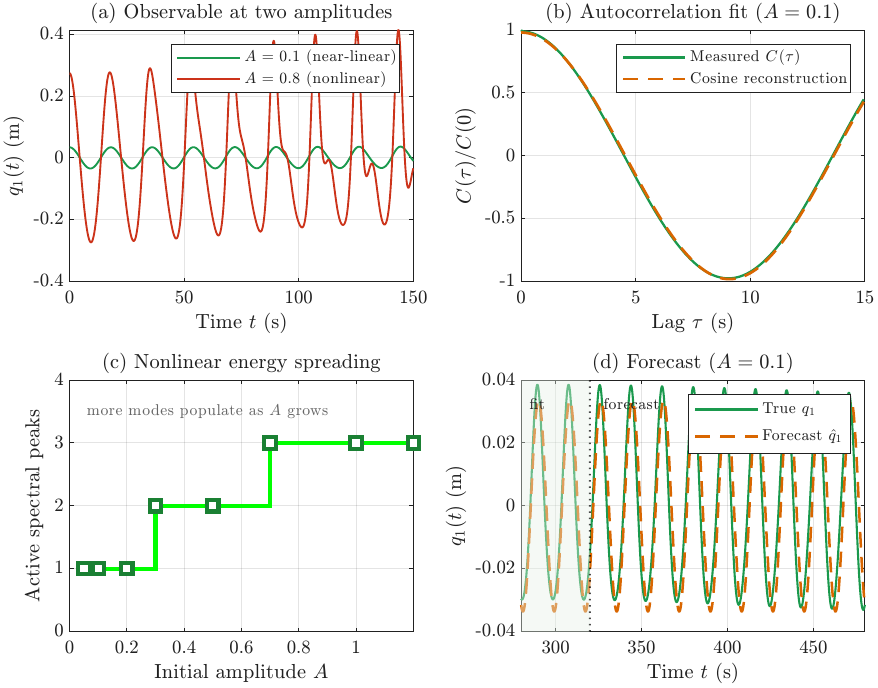}
\caption{\textbf{FPUT $\alpha$-chain.} The single-mode initial condition
spreads energy to other modes as amplitude grows: active spectral peaks
recovered from $q_1(t)$ rise from one ($A\le0.2$) to three ($A\ge0.7$), a
data-driven signature of nonlinear energy sharing; the small-amplitude
forecast (panel d) tracks the observable. 10\% additive Gaussian noise.
\label{fig:fput}}
\end{figure}

\begin{figure}[t]
\centering
\includegraphics[width=\linewidth]{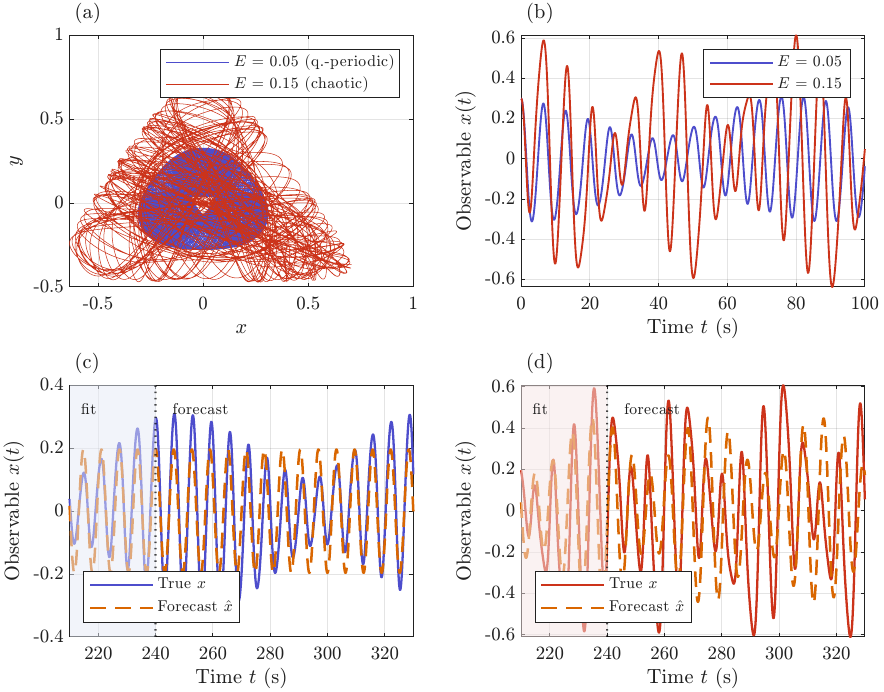}
\caption{\textbf{H\'enon-Heiles.} Out-of-sample forecast of $x(t)$: the
regular regime ($E=0.05$, panel c) is tracked by the recovered
single-frequency model, whereas the chaotic regime ($E=0.15$, panel d)
diverges after a short horizon---the forecast-breakdown horizon is itself a
diagnostic of chaos. 10\% additive Gaussian noise.
\label{fig:henon}}
\end{figure}

\section{Reactive coupling}
\label{sec:reactive}

Assumption~\ref{ass:A} lets the hidden modes push on the observable through
position.  That is the right picture for a hidden mass tied to the observable
by a spring, and it makes the residues positive semidefinite.  But some of the
systems we most want to reach do not work this way.  A sloshing tank, a
gyroscope, a tuned absorber---each reacts on the observable through its
\emph{acceleration}, not its position.  We call this reactive coupling, and it
changes the story in one important place: how the bare hidden frequencies are
recovered.

Eliminate the hidden coordinate as before, and the memory kernel is no longer a
sum of simple sections.  It becomes a general matrix \emph{reactance} function,
\begin{equation}
  \hat{\mat{K}}(s)=\mat{K}_\infty
    +\sum_{i=1}^{n_\hid}\frac{\mat{R}_i^{(0)}+\mat{R}_i^{(2)}s^2}{s^2+\omega_i^2}.
  \label{eq:reactive_kernel}
\end{equation}
Its poles still sit on the imaginary axis at the bare frequencies
$\{\omega_i\}$.  Two things are new.  The kernel no longer decays at high
frequency---the constant $\mat{K}_\infty$ is a static stiffness shift---and its
residues, read in the plain $\mat{R}_i/(s^2+\omega_i^2)$ basis, need not be
positive.  Passivity still holds, because the system is still conservative; but
the certificate is now Foster's reactance condition on $\hat{\mat{K}}$, not
$\mat{R}_i\succeq0$.  The non-negativity constraint in the Phase~2
program~\eqref{eq:nnls} must be relaxed accordingly.

The recovery question then splits cleanly in two.  The peaks of the observable
autocorrelation are the closed-loop frequencies $\{\Omega_j\}$, the roots of
$\det\!\big(s^2\mat{I}-\mat{G}-\hat{\mat{K}}(s)\big)=0$.  The bare frequencies
$\{\omega_i\}$ are something else: they are the poles of $\hat{\mat{K}}$, which
are the \emph{transmission zeros} of the observable receptance
\begin{equation}
  \mat{H}(s)=\big(s^2\mat{I}-\mat{G}-\hat{\mat{K}}(s)\big)^{-1}.
  \label{eq:receptance}
\end{equation}
In other words, a bare hidden frequency shows up as an \emph{anti-resonance} of
the observable, not a peak.  For position coupling the two coincide, and the
output-only method of Section~\ref{sec:algorithm} finds them.  For reactive
coupling they do not: the anti-resonance of the free-response autocorrelation is
displaced, and the deconvolution~\eqref{eq:kernel_from_corr} run on output-only
data returns a biased frequency.  Recovering $\{\omega_i\}$ without bias needs a
\emph{forced} experiment---an input--output measurement whose transmission zeros
are read off by rational FRF fitting.  This is a stronger data requirement than
the conservative output-only case, and it is the one that slosh- and gyro-type
problems actually demand.

\subsection{A worked example: cart and pendulum}
\label{sec:cartpend}

The smallest system that shows all of this is a spring-mounted cart carrying a
pendulum---the standard cart--pendulum of control texts~\cite[Ch.~3]{ogata2010modern},
and a minimal mechanical model of a single sloshing mode.  The cart
position $x$ is observed; the pendulum angle $\theta$ is hidden.  Linearizing
about $\theta=0$,
\begin{equation}
  \begin{bmatrix} M+m & ml\\ ml & ml^2\end{bmatrix}
  \begin{bmatrix}\ddot{x}\\ \ddot{\theta}\end{bmatrix}
  +\begin{bmatrix} k & 0\\ 0 & mgl\end{bmatrix}
  \begin{bmatrix}x\\ \theta\end{bmatrix}=\begin{bmatrix}F(t)\\0\end{bmatrix},
  \label{eq:cartpend}
\end{equation}
with cart mass $M$, bob mass $m$, length $l$, cart spring $k$, and forcing
$F(t)$.

\begin{figure}[t]
\centering
\includegraphics[width=0.62\linewidth]{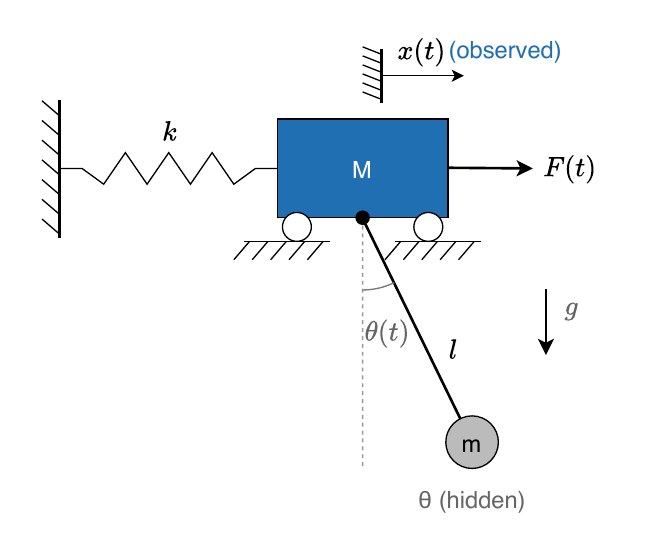}
\caption{\textbf{Cart--pendulum setup (slosh analog).} A spring-mounted cart
(mass $M$, spring $k$, position $x$) carries a pendulum (bob mass $m$, length
$l$, angle $\theta$).  The cart position $x$ is observed; the pendulum angle
$\theta$ is hidden.  The pendulum couples to the cart through its
acceleration---reactive coupling.\label{fig:cartpend_setup}}
\end{figure}  The pendulum enters through $\ddot{x}$, so the coupling is reactive.
Eliminating $\theta$ gives a single reactance section,
\begin{equation}
  \hat{K}(s)=\frac{mg}{l}\,\frac{s^2}{s^2+g/l}
    =\frac{mg}{l}-\frac{(mg/l)\,\omega_p^2}{s^2+\omega_p^2},
  \qquad \omega_p^2=\frac{g}{l}.
  \label{eq:cartpend_kernel}
\end{equation}
Its residue in the $R/(s^2+\omega^2)$ basis is $-mg^2/l^2$, which is negative---
outside Assumption~\ref{ass:A}---yet its pole still sits exactly at the bare
pendulum frequency $\omega_p=\sqrt{g/l}$.

Take $M=1$, $m=0.5$, $l=1$, $k=15$, and $g=9.81\,\mathrm{m/s^2}$, so
$\omega_p=3.13\,\mathrm{rad/s}$.  The cart's autocorrelation peaks at the
closed-loop frequencies $2.51$ and $4.84\,\mathrm{rad/s}$.  The pendulum
frequency is not among them; it falls in the gap between them
(Fig.~\ref{fig:cartpend}b).  Run the output-only
deconvolution~\eqref{eq:kernel_from_corr} on the free response and it returns a
biased value---and, being excitation-dependent, an unstable one.  The forced
receptance tells a cleaner story: it has a transmission zero exactly at
$\omega_p$, and a sine-dwell sweep of the cart recovers that anti-resonance
under $10\%$ measurement noise (Fig.~\ref{fig:cartpend}a).  Because the cart
barely moves near the anti-resonance, the frequency is best read from a rational
fit of the whole FRF, not from the raw notch.

We state the boundary plainly.  For reactively coupled hidden modes PRISM still
recovers the bare frequency, but only from a forced input--output experiment and
with a Foster kernel model in place of the positive-residue one.  The
conservative output-only route of Benchmarks~1--3 is not enough.  This is the
honest scope of the method---and the case that makes the slosh and gyroscope
applications rigorous rather than aspirational.

\begin{figure}[t]
\centering
\includegraphics[width=\linewidth]{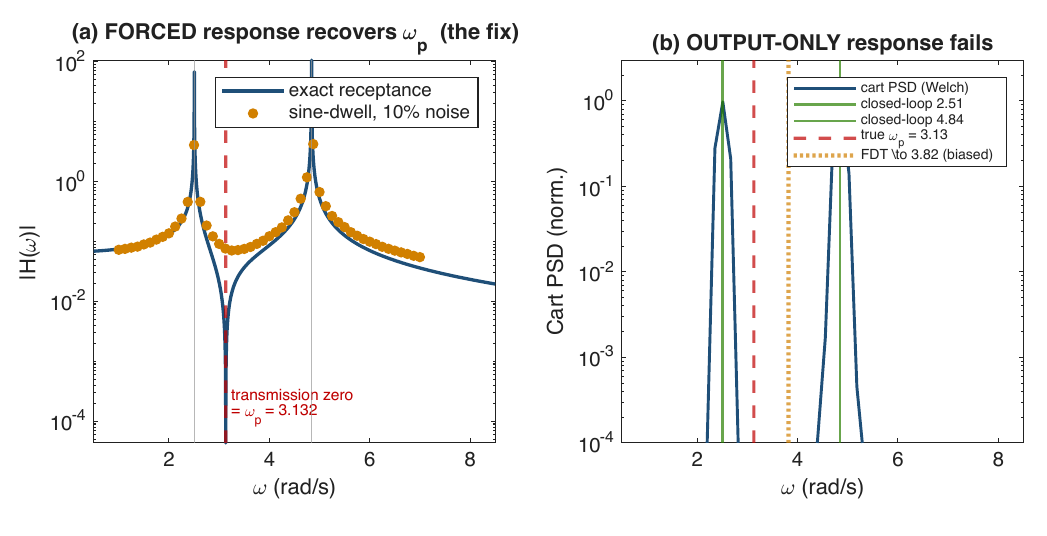}
\caption{Cart--pendulum (slosh analog), reactive coupling, $g=9.81$; hidden
  $\omega_p=\sqrt{g/l}=3.13\,\mathrm{rad/s}$.  \textbf{(a)} forced receptance
  $|H(\omega)|$ with sine-dwell samples ($10\%$ noise); \textbf{(b)} output-only
  cart spectrum.  See text for details.\label{fig:cartpend}}
\end{figure}

\section{Comparison with existing methods}
\label{sec:comparison}

The deliverable of PRISM is a \emph{governing equation}: the closed
second-order equation of motion for the observable,
Eq.~\eqref{eq:MZ}, equivalently the minimal Hamiltonian
Eq.~\eqref{eq:reconstructed_ham} that generates it.  The relevant
comparison is therefore with other methods that discover governing
equations from data; we group them by what they require and what they
return (Table~\ref{tab:comparison}).  Pure spectral estimators
(Prony, ESPRIT, matrix pencil~\cite{hua1990matrix}) are \emph{not} in
this class: they return a list of frequencies and per-channel
amplitudes, not an equation of motion, and so cannot be integrated
forward, perturbed, or interrogated for the observed--hidden coupling.
On their own task---extracting a sum of sinusoids---they are mature and
accurate (on the Benchmark~1 data ESPRIT and matrix pencil match PRISM's
frequencies to within $0.01\%$), but recovering frequencies is a
\emph{sub-step} of PRISM, not its output, so we do not compare against
them as equation-discovery methods.

\begin{table}[ht!]
\centering
\caption{Comparison of governing-equation discovery methods. Deliverable
is an equation of motion, not a spectrum; pure spectral estimators
(ESPRIT, etc.) are excluded as they return no dynamical model (see text).
PRISM uniquely recovers an interpretable, passive, physical model from
partial observations with no state reconstruction or trained network.
Key: \checkmark\ yes, $\times$ no, $\sim$ approximate/optional,
``--'' not applicable. DD = deep delay.
\label{tab:comparison}}
\setlength{\tabcolsep}{6pt}
\renewcommand{\arraystretch}{1.15}
\footnotesize
\begin{tabular}{lcccc}
\toprule
                          & SINDy & Koop. & DD & \textbf{PRISM} \\
\midrule
Equation of motion        & \checkmark & $\sim$ & \checkmark & \textbf{\checkmark} \\
Partial observation       & $\times$  & $\times$ & \checkmark & \textbf{\checkmark} \\
No reconstruction         & --         & --     & $\times$  & \textbf{\checkmark} \\
No delay param.\ $\tau$    & --         & --     & $\times$  & \textbf{\checkmark} \\
No trained network        & \checkmark & $\sim$ & $\times$  & \textbf{\checkmark} \\
Interpretable model       & \checkmark & $\times$ & $\times$ & \textbf{\checkmark} \\
Passivity guarantee       & $\times$  & $\times$ & $\times$ & \textbf{\checkmark} \\
\bottomrule
\end{tabular}
\end{table}

Within this class, PRISM occupies a position no existing method
reaches.  SINDy and Koopman/EDMD recover dynamics but require the
\emph{full} state; on a partial observation they must first reconstruct
the hidden coordinates, which is the very step PRISM avoids.  Deep delay
methods do operate on partial observations and return a forward model,
but in learned delay coordinates through a trained network: the result
is a black-box predictor, not the physical Hamiltonian, and carries no
conservation or passivity guarantee.  PRISM instead returns the minimal
physical Hamiltonian directly---interpretable coupling vectors between
observed and hidden coordinates, with passivity enforced by the
$\mat{R}_i\succeq0$ structure---and, because it is an equation of motion,
it can be integrated forward: the forecast panels of
Figs.~\ref{fig:chain} and \ref{fig:henon} are precisely this test, a
recovered model propagated beyond its fitting window, which a spectral
description cannot produce.

We are explicit about the scope of this advantage.  PRISM is superior
\emph{for the problem it targets}---partially observed conservative
systems with harmonic hidden modes (Assumption~\ref{ass:A})---and
not in general.  SINDy addresses arbitrary nonlinearities outside
PRISM's linear-coupling structure; neural methods scale to higher
dimensions than the rational-kernel fit; and where the hidden dynamics
are themselves strongly nonlinear, PRISM recovers an effective
linearization rather than the exact model (Benchmarks~2 and 3).  The
claim is not that PRISM dominates equation discovery, but that for
conservative partially observed systems it uniquely returns an
interpretable, passive, physical equation of motion without
reconstruction, embedding, or a trained network.

\section{Applications}
\label{sec:applications}

PRISM applies to any system whose full state is
Hamiltonian and whose observable is a low-dimensional projection.  In
Hamiltonian fluid~\cite{morrison1998hamiltonian} and
plasma~\cite{freidberg2014ideal} models, the residues track coherent or
tearing-mode frequencies from a single accessible field probe; in open
quantum systems~\cite{caldeira1983quantum}, the kernel coincides with the
bath spectral density (strong-coupling recovery requires inverting the
system--bath hybridization, left to future work).  The same structure
governs engineering settings, though not always through the output-only route.
Fuel slosh and gyroscopic coupling are \emph{reactive}
(Section~\ref{sec:reactive}): a spacecraft slosh
mode~\cite{abramson1966dynamic} appears as a transmission zero of the
gyro-rate receptance, recovered under forced maneuvering rather than from the
free-response autocorrelation.  In aeroelastic flutter
prediction~\cite{bisplinghoff1955aeroelasticity}, the index
$\mu(V)=\min_i\lambda_\mathrm{min}(\mat{R}_i(V))\to0^+$ as airspeed approaches
$V_F$ provides a real-time precursor.  Extensions to
dissipative hidden modes and time-varying parameters are discussed
in the appendices.


\section{Conclusion}
\label{sec:conclusion}

PRISM exploits the Hamiltonian structure of hidden physics to discover
governing equations without reconstructing the hidden state.  Three
pillars---the exact MZ projection, the lossless positive-real kernel-structure
theorem, and sparse identification of the pole--residue parameters---make the
reduced model passive by construction, and, where the microscopic Hamiltonian
is identifiable, unique and minimal.

Validation on three benchmarks spanning a linear chain, a nonlinear
lattice, and a chaotic Hamiltonian system
demonstrates robust frequency recovery under realistic noise, and
coupling (residue) recovery under the equipartition/observability
condition of Theorem~\ref{thm:identifiability}.
The positive-semidefinite residue structure is the single object that
enforces the physical guarantees; that it can be imposed by convex
(non-negative least-squares, or semidefinite) programming is what makes
the method practical.

The central result is a governing equation, not a spectrum: from the output
autocorrelation of a partially observed conservative system we recover a
closed, passive, interpretable governing equation for the observed dynamics---the
Mori--Zwanzig memory-kernel equation---whose closed-loop normal-mode
frequencies are the peaks of that autocorrelation and which forecasts the
observable beyond its fitting window.  The bare hidden Hamiltonian follows as a
conditional bonus: through the memory-kernel
deconvolution~\eqref{eq:kernel_from_corr} when the observability Gramian
$\mat{W}_o\succ0$ and the data sample an equipartitioned measure (position
coupling), or through a forced input--output experiment when the coupling is
reactive.  What is cheap---the reduced governing equation---and what is
expensive---the microscopic Hamiltonian---are thereby separated cleanly.
The same structure recurs across mechanics, fluid and plasma physics, and
molecular dynamics, which the framework treats on a common footing.

\appendix

\section{Proof of Theorem~\ref{thm:kernel_structure} (Kernel Structure)}
\label{app:proof1}

We derive the
equations of motion directly from the full Hamiltonian~\eqref{eq:full_ham},
\begin{equation}
  H = H_\obs(\vect{x},\dot{\vect{x}})
    + \sum_{i=1}^{n_\hid}\!\left(\tfrac{1}{2}p_i^2+\tfrac{1}{2}\omega_i^2 q_i^2\right)
    + \sum_{i=1}^{n_\hid} \vect{x}^\top \vect{c}_i\,q_i,
  \label{eq:SHfull}
\end{equation}
with $H_\obs = \tfrac{1}{2}\dot{\vect{x}}^\top \dot{\vect{x}} + V(\vect{x})$
and $M=I$.  Hamilton's equations for the observable position
$\vect{x}=\vect{q}_\obs$ give $\dot{\vect{x}}=\partial H/\partial\vect{p}_\obs$
and
\begin{equation}
  \dot{\vect{p}}_\obs = -\frac{\partial H}{\partial \vect{x}}
    = -\nabla_{\vect{x}} V(\vect{x})
    - \sum_{i=1}^{n_\hid} \vect{c}_i\,q_i.
  \label{eq:Spobs_dot}
\end{equation}
Since $\dot{\vect{p}}_\obs = \ddot{\vect{x}}$ and
$\vect{g}\equiv -\nabla_{\vect{x}}V$, this gives
\begin{equation}
  \ddot{\vect{x}}(t) = \vect{g}(\vect{x},\dot{\vect{x}})
    - \sum_{i=1}^{n_\hid}\vect{c}_i\,q_i(t).
  \label{eq:Sobs}
\end{equation}
Hamilton's equations for the $i$-th hidden mode give
$\dot{q}_i = \partial H/\partial p_i = p_i$ and
$\dot{p}_i = -\partial H/\partial q_i = -\omega_i^2 q_i - \vect{c}_i^\top\vect{x}$;
differentiating $\dot q_i=p_i$ once more yields
\begin{equation}
  \ddot{q}_i + \omega_i^2\,q_i = -\vect{c}_i^\top\vect{x}(t).
  \label{eq:Shid}
\end{equation}
Equations~\eqref{eq:Sobs}--\eqref{eq:Shid} are the starting point.

\subsection{Variation of parameters}
The general solution of the forced
oscillator~\eqref{eq:Shid} is
\begin{align}
  q_i(t) &= \underbrace{q_i(0)\cos\omega_i t
    + \frac{\dot{q}_i(0)}{\omega_i}\sin\omega_i t}_{r_i(t)} \nonumber\\
    &\quad - \frac{\vect{c}_i^\top}{\omega_i}
      \int_0^t \sin\bigl(\omega_i(t-s)\bigr)\vect{x}(s)\,ds.
  \label{eq:Svop}
\end{align}

\subsection{Substitution}

Inserting~\eqref{eq:Svop} into~\eqref{eq:Sobs}:
\begin{align}
  \ddot{\vect{x}}(t)
  &= \vect{g} + \int_0^t
    \underbrace{\sum_{i}\frac{\vect{c}_i\vect{c}_i^\top}{\omega_i}
    \sin\bigl(\omega_i(t-s)\bigr)}_{=\,\mat{K}(t-s)}
    \vect{x}(s)\,ds
    - \underbrace{\sum_i\vect{c}_i r_i(t)}_{=\,-\vect{F}(t)}.
\end{align}
This is exactly Eq.~\eqref{eq:MZ} with:
\begin{equation}
  \mat{K}(t) = \sum_{i=1}^{n_\hid}
    \frac{\mat{R}_i}{\omega_i}\sin(\omega_i t),
  \qquad \mat{R}_i = \vect{c}_i\vect{c}_i^\top \succeq 0.
  \label{eq:Skernel_time}
\end{equation}

\subsection{Laplace transform}

Using $\mathcal{L}\{\sin(\omega_i t)\}=\omega_i/(s^2+\omega_i^2)$:
\begin{equation}
  \hat{\mat{K}}(s) = \sum_{i=1}^{n_\hid}
    \frac{\mat{R}_i}{s^2+\omega_i^2}.
\end{equation}
Poles at $s=\pm \iu\omega_i$ lie on the imaginary axis (no dissipation).

\subsection{Lossless passive character}
The residue at $s=+\iu\omega_i$ is:
\begin{equation}
  \mathrm{Res}_{s=\iu\omega_i}\hat{\mat{K}}(s)
  = \frac{\mat{R}_i}{2\iu\omega_i}, \qquad \mat{R}_i\succeq 0.
\end{equation}
Two properties follow.  First, all poles lie on the imaginary axis
$s=\pm\iu\omega_i$ (no real part), so the kernel is \emph{lossless}: it
dissipates no energy.  Second, the residues are positive semidefinite, so
the spectral density
\begin{equation}
  \mat{S}(\omega)=\sum_i\mat{R}_i\,\delta(\omega-\omega_i)
\end{equation}
is a positive matrix measure.  By the matrix-valued Bochner
theorem~\cite{bochner1932vorlesungen} this is equivalent to a
positive-semidefinite autocorrelation, certifying that the kernel is
\emph{passive}: energy is stored and exchanged among the hidden modes but
never generated.  Together, the kernel is lossless and passive --- the
reactance (lossless positive-real) case~\cite{anderson1967network,youla1961modern}
--- consistent with its conservative origin.
\hfill$\square$

\section{Proof of Theorem~\ref{thm:identifiability} (Identifiability)}
\label{app:proof2}

Define
$\mat{C}(\tau)=\Expect{\vect{x}(t)\vect{x}(t+\tau)^\top}$ (independent of
$t$ by stationarity).
To obtain its equation of motion, right-multiply the Mori--Zwanzig
equation~\eqref{eq:MZ} by $\vect{x}(0)^\top$ and take the ensemble average
$\Expect{\cdot}$ over the (equipartitioned) initial measure.  Linearizing
the observable self-force about the operating point $\vect{x}_0$ (an
equilibrium of the observable subsystem),
$\vect{g}(\vect{x})=-\nabla_{\vect{x}}V \approx \mat{G}\,\vect{x}$ with the
observable stiffness $\mat{G}=-\nabla^2_{\vect{x}}V|_{\vect{x}_0}$ (exact
when $V$ is quadratic, as in the linear benchmark; $\mat{G}=\mat{0}$ for
free or gyroscopic observables).  The average is linear and commutes with
the $\tau$-derivative and the convolution, so
\begin{align}
  \Expect{\ddot{\vect{x}}(\tau)\vect{x}(0)^\top}
  &= \mat{G}\,\Expect{\vect{x}(\tau)\vect{x}(0)^\top}
  + \int_0^\tau\!\mat{K}(\tau-s)\Expect{\vect{x}(s)\vect{x}(0)^\top}ds \nonumber\\
  &\quad + \Expect{\vect{F}(\tau)\vect{x}(0)^\top}.
\end{align}
Identifying $\Expect{\vect{x}(\tau)\vect{x}(0)^\top}=\mat{C}(\tau)$,
$\Expect{\ddot{\vect{x}}(\tau)\vect{x}(0)^\top}=\ddot{\mat{C}}(\tau)$
(the derivative acts on $\tau$ alone), and using
$\Expect{\vect{F}(\tau)\vect{x}(0)^\top}=0$ (the fluctuation force depends
only on the initial hidden conditions, uncorrelated with $\vect{x}(0)$ by
assumption), gives
\begin{equation}
  \ddot{\mat{C}}(\tau)=\mat{G}\mat{C}(\tau)
    +\int_0^\tau\mat{K}(\tau-s)\mat{C}(s)\,ds.
  \label{eq:Scorr}
\end{equation}

\subsection{Laplace inversion}

Taking the Laplace transform of~\eqref{eq:Scorr} and solving for
$\hat{\mat{K}}(s)$:
\begin{equation}
  \hat{\mat{K}}(s) = \bigl[s^2\hat{\mat{C}}(s) - s\mat{C}(0)
    - \mat{\Lambda}\mat{C}(0) - \mat{G}\hat{\mat{C}}(s)\bigr]
    \hat{\mat{C}}^{-1}(s).
  \label{eq:Skinv}
\end{equation}

\begin{lemma}
$\hat{\mat{C}}(s)$ is invertible for $\mathrm{Re}(s)>0$ iff
$\mat{W}_o\succ0$.
\end{lemma}
\begin{proof}
By Parseval's theorem, $\mathrm{tr}(\mat{W}_o)=\int_0^\infty
\|\mat{C}(\tau)\|_F^2\,d\tau>0$ iff the observable orbit is not confined
to a proper invariant subspace, which is equivalent to
$\hat{\mat{C}}(s)$ having full column rank in the right half-plane.
\end{proof}

Eq.~\eqref{eq:Skinv} and the Lemma together show $\hat{\mat{K}}(s)$
is uniquely determined by $\mat{C}(\tau)$ when $\mat{W}_o\succ0$.
Uniqueness of the partial-fraction decomposition follows from the
uniqueness of Laurent coefficients at isolated poles of rational
matrices~\cite{gantmacher1959}.\hfill$\square$

\subsection{Sufficient data length}

For $n_\hid$ hidden modes with minimum modal separation
$\Delta\omega_\mathrm{min}=\min_{i\neq j}|\omega_i-\omega_j|$,
identifiability requires $T\geq\pi/\Delta\omega_\mathrm{min}$.

\subsection{Modal decomposition and the role of equipartition}
The
preceding inversion recovers the pole-residue pair $(\omega_i,\mat{R}_i)$,
but the conditions under which the frequency, the residue direction, and
the residue magnitude are recoverable differ, and the distinction is made
precise by expanding $\mat{C}(\tau)$ in normal modes.  Writing the
observable as $\vect{x}(t)=\sum_j \vect{v}_j\,a_j(t)$, where $a_j$ is the
$j$-th \emph{closed-loop} normal-mode amplitude (a root of
$\det(s^2\mat{I}-\mat{G}-\hat{\mat{K}}(s))=0$, in general distinct from any
bare hidden frequency $\omega_i$ of Theorem~\ref{thm:kernel_structure}) and
$\vect{v}_j$ its shape on the observed
coordinates, and using that distinct modes are uncorrelated under a
stationary measure,
$\Expect{a_j(t+\tau)a_l(t)}=\delta_{jl}\Expect{a_j(t+\tau)a_j(t)}$, gives
\begin{equation}
  \mat{C}(\tau)=\sum_j \vect{v}_j\vect{v}_j^\top
    \,\Expect{a_j(t+\tau)\,a_j(t)}.
\end{equation}
Each mode is an undamped oscillator, whose stationary autocorrelation is
$\Expect{a_j(t+\tau)a_j(t)}=\Expect{a_j^2}\cos\Omega_j\tau$; by the virial
theorem $\Expect{E_j}=\Omega_j^2\Expect{a_j^2}$, so
\begin{equation}
  \mat{C}(\tau)=\sum_j \mat{R}_j\,
    \frac{\Expect{E_j}}{\Omega_j^2}\,\cos\Omega_j\tau,
  \qquad \mat{R}_j=\vect{v}_j\vect{v}_j^\top.
  \label{eq:modal}
\end{equation}
The single-mode relations
$\Expect{a_j\dot{a}_j}=0$ and $\Expect{\dot{a}_j^2}=\Omega_j^2\Expect{a_j^2}$
used here are not additional hypotheses: they follow from averaging an
oscillation $a_j(t)=A_j\cos(\Omega_j t+\phi_j)$ over its phase (equivalently,
over time along a sufficiently long record, $T\geq\pi/\Delta\omega_\mathrm{min}$),
and hold for any modal amplitude $A_j$.  The modal energies $\Expect{E_j}$
therefore enter only through the residue magnitudes.  Three consequences
follow.  (i)~The $\{\Omega_j\}$ are the closed-loop cosine
frequencies, recoverable from any persistently exciting measurement
regardless of the modal energies.  (ii)~Each term's matrix \emph{direction}
is $\vect{v}_j\vect{v}_j^\top=\mat{R}_j$, independent of $\Expect{E_j}$, so
the residue \emph{structure} (coupling direction) is likewise recoverable
from any such measurement.  (iii)~The residue \emph{magnitudes} carry the
factor $\Expect{E_j}/\Omega_j^2$; after de-tilting by $\Omega_j^2$ the
recovered amplitude is $\mat{R}_j\Expect{E_j}$, weighted by the modal
energy.  The relative magnitudes across modes are therefore correct only
under \emph{equipartition}, $\Expect{E_j}\equiv E$, which makes the
weighting uniform and fixes all residues up to a single global constant
(set by $\mat{C}(0)$).  Equipartition is thus required only for
residue-magnitude recovery; frequencies and coupling directions need only
persistent excitation.  The bare $\{\omega_i\}$ of
Theorem~\ref{thm:kernel_structure} are recovered from
$\{\Omega_j,\mat{R}_j\}$ only via the deconvolution of
\ref{app:fdt}, not from this modal decomposition directly.

\section{Proof of Reconstruction Uniqueness}
\label{app:proof3}

Each
$\mat{R}_i\succeq0$ (from Theorem~\ref{thm:kernel_structure}), so eigendecompose:
$\mat{R}_i=\mat{U}_i\,\mathrm{diag}(\lambda_{i,k})\,\mat{U}_i^\top$,
$\lambda_{i,k}\geq0$.
Then $\vect{c}_i=\sqrt{\lambda_{i,1}}\,\vect{u}_{i,1}$ (rank-1 case)
recovers the coupling vector up to sign convention.

\subsection{Uniqueness}

Suppose two Hamiltonians $\Ham^{(1)}$ and $\Ham^{(2)}$ produce the
same $\mat{C}(\tau)$.
By~\eqref{eq:Skinv} they produce the same $\hat{\mat{K}}(s)$.
The partial-fraction representations must then coincide
(uniqueness of Laurent coefficients~\cite{gantmacher1959}), so
$N^{(1)}=N^{(2)}$, $\omega_i^{(1)}=\omega_{\sigma(i)}^{(2)}$,
and $\mat{R}_i^{(1)}=\mat{R}_{\sigma(i)}^{(2)}$ for a permutation $\sigma$.
The Hamiltonians are identical up to mode relabeling and the
$\pm\vect{c}_i$ sign ambiguity (which leaves $\Ham$ unchanged).

\subsection{Bochner uniqueness}

The matrix-valued Bochner theorem~\cite{bochner1932vorlesungen} states
that any positive-semidefinite matrix-valued function $\mat{C}(\tau)$
has a unique spectral representation:
$\mat{C}(\tau)=\int e^{j\omega \tau}d\mat{F}(\omega)$,
where $\mat{F}(\omega)$ is a positive-semidefinite matrix measure.
For our system $d\mat{F}$ is supported on $\{\pm\omega_i\}$,
establishing a one-to-one correspondence between $\mat{C}(\tau)$ and
the Hamiltonian parameters.\hfill$\square$

\section{Fluctuation--Dissipation Identity}
\label{app:fdt}

The exact inversion
used in Section~\ref{sec:identifiability} follows from the generalized fluctuation--dissipation
relation.

\begin{lemma}[Generalized FDT]
Assume (i) the self-force is linearized about the operating point,
$\vect{g}(\vect{x})\approx\mat{G}\vect{x}$ with
$\mat{G}=-\nabla^2_{\vect{x}}V|_{\vect{x}_0}$ (exact for quadratic $V$), and
(ii) the hidden initial conditions are uncorrelated with the observable at
equal time, $\Expect{q_i(0)\vect{x}(0)^\top}=\Expect{\dot{q}_i(0)\vect{x}(0)^\top}=0$
(equivalently $\Expect{\vect{F}(\tau)\vect{x}(0)^\top}=0$, since the
fluctuation force is a linear combination of the propagated hidden initial
conditions).  Then the kernel is recovered from the autocorrelation by the
deconvolution
\begin{align}
  \hat{\mat{K}}(s) &= \bigl[s^2\hat{\mat{C}}(s) - s\mat{C}(0)
    - \mat{\Lambda}\mat{C}(0) - \mat{G}\hat{\mat{C}}(s)\bigr]
    \hat{\mat{C}}^{-1}(s), \label{eq:Sfdt}\\
  \mat{K}(\tau) &= \mathcal{L}^{-1}\{\hat{\mat{K}}(s)\},
  \label{eq:Sfdt_inv}
\end{align}
with streaming matrix $\mat{\Lambda}=\dot{\mat{C}}(0)\mat{C}^{-1}(0)$.
\end{lemma}
\begin{proof}
Under (i)--(ii), $\mat{C}(\tau)$ obeys the linearized autocorrelation
equation~\eqref{eq:Scorr}; condition~(ii) ensures no equal-time
observable--hidden correlation survives on the right-hand side, so the
relation closes in $\mat{C}$.  Taking the Laplace transform
of~\eqref{eq:Scorr}, with $\mathcal{L}\{\ddot{\mat{C}}\}=s^2\hat{\mat{C}}(s)
-s\mat{C}(0)-\dot{\mat{C}}(0)$ and the convolution theorem
$\mathcal{L}\{\mat{K}*\mat{C}\}=\hat{\mat{K}}(s)\hat{\mat{C}}(s)$, and
solving for $\hat{\mat{K}}(s)$ gives~\eqref{eq:Sfdt}.  Recovery of
$\mat{K}(\tau)$ is therefore a \emph{deconvolution} (division by
$\hat{\mat{C}}(s)$ in the transform domain), not a pointwise operation; the
factor $\hat{\mat{C}}^{-1}(s)$ cannot be replaced by the constant
$\mat{C}^{-1}(0)$.  In practice the rational fit of the main-text algorithm
performs this inversion in pole-residue form, avoiding explicit transform
inversion.
\end{proof}


\section*{Data availability}
The simulation data and all scripts
required to reproduce the three benchmark results are publicly
available at \url{https://github.com/maayoubi/PRISM} and archived at
DOI: 10.5281/zenodo.20552805.

\section*{Declaration of interests}
The author declares no competing financial or non-financial
interests.

\bibliographystyle{elsarticle-num}
\bibliography{PRISM_refs}

\end{document}